\documentclass{article}
\usepackage{hyperref}
\usepackage{amscd,amsmath,amsthm,amssymb}
\usepackage{enumerate,varioref}
\usepackage{epsfig}
\usepackage{graphicx}
\usepackage{mathtools}
\usepackage{tikz}
\newtheorem{thm}{Theorem}

\newtheorem{lem}[thm]{Lemma}

\theoremstyle{definition}
\newtheorem{defn}[thm]{Definition}
\theoremstyle{remark}

\newtheorem{rem}[thm]{Remark}

\newcommand{\Z}{\mathbb{Z}}

\newcommand{\half}{\frac{1}{2}}

\title{The Shape Metric for Clustering Algorithms}
\author{Clark Alexander \\Sofya Akhmametyeva \\
Nousot Applied Research Group\footnote{Additional NARG members, R. Jachi} \\ email \href{mailto:clark@nousot.com}{clark@nousot.com}}

\begin{document}
\maketitle

\begin{abstract}
We construct a method by which we can calculate the precision with which an algorithm identifies the shape of a cluster.  We present our results for several well known clustering algorithms and suggest ways to improve performance for newer algorithms. 
\end{abstract}

\section{Introduction}

One amongst many fundamental techniques in machine learning is data clustering.  One may argue that clusters are inventions, or that no matter which clustering algorithm one uses, one always loses information.  While this may be so, the utility of clustering algorithms should be not underestimated.  In this article the authors attempt to define what is meant by ``the shape of data" and precisely score this.  Many survey works (cf \cite{Cluster-Book}\cite{Bench}\cite{HDBSCAN}) give qualitative measures of strengths and weaknesses of individual algorithms.  Among the most intriguing qualitative feature to the authors is ``shape."  Several of the survey works previously mentioned discuss whether or not an algorithm is able to discern an arbitrary shape.  Algorithms such as HDBSCAN and Wavecluster are able to discern shapes well.  Others such as k-means and its extensions and BIRCH struggle a little more with shape.  Our fundamental question is, ``which algorithm discerns shape the best?"  We attempt to answer this question by putting a pseudo-metric on shapes.  The reason for having a Pseudo-metric, as we shall soon see is that it is possible for two different shapes to have the same shape complexity.  In this case we can partition shapes into equivalence classes of shape complexity and thus our pseudo-metric becomes a true metric on the space of equivalence classes. \\

\par This article proceeds as follows.  In section 2, we discuss how to extract the ``boundary" from a cluster.  In fact we shall extract two boundaries.  The first is the boundary of edges of a triangulated graph, the second is the actual set of points making a 
set of cyclic graphs.  We shall use these notions in concert to extract what we call a ``boundary."  \\
\par In section three, we use the notion of boundary defined in section 2 to fit a curve (or set of curves) to the boundary of our cluster.  This curve will turn out to be a NURBS curve, which is a close relative to the B\`{e}zier curves. By adjusting the weights of the NURBS curve we can affect the local and global curvatures of our set of curves.  We will also discuss how to extend this technique to NURBS surfaces for good visualization in two and three dimensions, and show how the mathematical technique can be used in higher dimensional analysis.  Fortunately the extension of NURBS curves to higher dimensions allows us to separate variables in each dimension.  We discuss the possibility of dimensional reduction by using variables dependent on each other in higher dimensions.\\
\par In section 4, having fit a curve (or higher dimensional analog) to our boundary, we now calculate its squared curvature along the boundary (or all connected components of the boundary).  In higher dimensions we shall calculate the sum of squares of all principal components of curvature.  For example, on surfaces we compute

\begin{equation}
\int_{\bigcup_i \partial C_i} (\kappa^2_1 + \kappa^2_2) d\sigma
\end{equation}

Much of the present work will be restricted to two and three dimensions for the sake of our ability to visualize the results.  Mathematically, this can extended to arbitrarily high dimensions, but the scale must be considered more carefully.  It is the suggestion of the authors that one first normalize the volume of dataspace to have unit volume.  Thus we can mostly avoid divergent integrals for very large, but not very curvy cluster boundaries.

\section{Boundary Extraction}

In order to fit a curve to the boundary of a dataset we must first extract the boundary.  In this case we follow the approach given in \cite{LEC}.  Their method is broken into several pieces on its own.  In a nutshell, the method of Lee and Estivill-Castro is described as follows:
\begin{itemize}
\item[1:] Compute the Delaunay Triangulation on the entire dataset.\\
\item[2:] Compute the mean and variance/Standard deviation of all edges lengths (Euclidean norm).\\
\item[3:] At each vertex $v_i$ sort all edges into short, long, and other.  These classes are determined by short lengths being more than $m$ standard deviations below the mean, long edges $m$ standard deviations above the mean, and other being all other edges.\\
\item[4:] Vertices sharing short edges are immediately in the same cluster.  Vertices sharing long edges are in separate clusters.  For the other edges one checks whether two points are in the same cluster mutually shared neighbors in the triangulation graph.\\
\item[5:] Remove all non-boundary edges by checking whether an edge is in more than one triangle.\\
\item[6:] Orient the boundary edges consistently.  
\end{itemize} 

\begin{rem}
It should be noted that a consistent orientation can be found in two dimensions, but not necessarily in three or higher.  We avoid this problem by computing curvature which is a local property and does not rely on the orientability of the underlying manifold.  Additionally, if we have self-intersecting data from dimension reduction we know that there is some higher dimension in which this can be embedded so that it is not self-intersecting.  Afterall, we know that we can unknot knots in 4 dimensions, and untie general higher dimensional knots in codimension 2.  
\end{rem}

It is also important to note that on its own the Delaunay triangulation will not give nonconvex boundaries.  This is why step 4 in the above algorithm is crucial.

Once we have extracted the clusters via the Delaunay triangulation we can see very clearly that the boundary edges are those which are part of exactly one triangle.  For example, consider
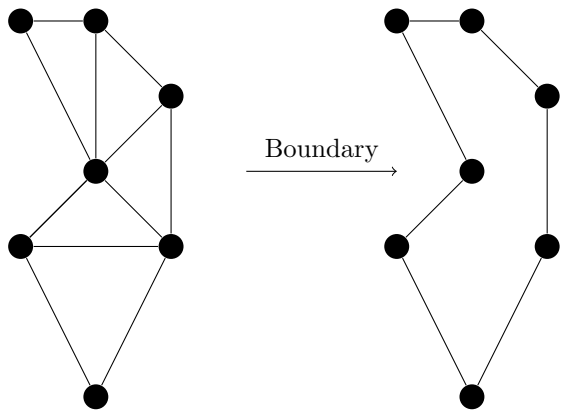
\[
\begin{tikzpicture}
\node[circle,fill](1) at (0,0){};
\node[circle,fill](2) at (0,2){};
\node[circle,fill](3) at (1,1){};
\node[circle,fill](4) at (1,-1){};
\node[circle,fill](5) at (0,-3){};
\node[circle,fill](6) at (-1,-1){};
\node[circle,fill](7) at (-1,2){};
\node[circle,fill](8) at (5,0){};
\node[circle,fill](9) at (5,2){};
\node[circle,fill](10) at (6,1){};
\node[circle,fill](11) at (6,-1){};
\node[circle,fill](12) at (5,-3){};
\node[circle,fill](13) at (4,-1){};
\node[circle,fill](14) at (4,2){};
\node[](15) at (2,0){};
\node[](16) at (4,0){};
\draw[->] (2,0) -- node[above] { Boundary } (4,0);
\foreach \from/\to in {1/2,1/3,1/4,1/6,1/6,2/3,3/4,4/5,5/6,4/6,1/7,2/7,
9/10,10/11,11/12,12/13,8/13,8/14,9/14}
 \draw (\from) -- (\to);
 
\end{tikzpicture}
\]

We can see clearly that the four inner edges which belong to more than one triangle have been removed.  Now consider the case in a higher number of dimensions.  In three dimensions, the triangles are now faces of tetrahedra.  We wish not only to extract the edges and boundary points, but also the triangular faces.  Our criterion, again is that each face must belong to exactly one tetrahedron.  This is its essence is a form of discrete integration.  In two dimensions we can think of orienting all the triangles clockwise (or counterclockwise as one prefers) and signing the edges.  Once we add all edges the edges which are in two triangles cancel by having a $+1$ marker for one triangle and a $-1$ marker for the second.  All the leftover edges are marked identically.  This is the spirit of the higher dimensional boundary extraction, but without the guarantee of orientable manifolds, we cannot consistently mark the simplicial complex to uniformly cancel all interior simplices.  For example, $\mathbb{RP}^n$ is unorientable and we can see this by its top homology group of $\Z/2$ which gives us the explicit obstruction to orientation.  Nonetheless, with spatial data we can still compute boundaries without relying on orientation, and thus we can still compute curvature.

\section{Curve Fitting}

For the moment let us continue in our discussion in two dimensions.  Once we have extracted our boundary, we need to fit a curve to it.  In principal we can simply allow the angular boundary of the Delaunay triangulation to be our shape, but it is impossible to distinguish in a numerical sense how well an algorithm can discern an arbitrary shape.  This question will arise naturally in image processing and, in a nearly identical fashion, audio signal processing.  The boundary extraction procedure from the previous section will tell us that every cluster came from Picasso's cubist period.  Obviously this is not the case.  Facial recognition software could not possibly be accurate if it required every person to have an angular face with no smooth curves.\\

\par For these reasons we shall use the points on the boundary and fit a curve to it.  In two dimensions, as mentioned, we can always consistently orient our boundary, and thus we can fit a curve by tracing near the boundary points \emph{in order}.  In higher dimensions, this cannot work as we may not have a consistent orientation to our boundary.  Additionally, and practically speaking, the more immediate obstruction is that we cannot impose a well-ordering in higher dimensions.  The best we can do is to ask for a partial order of a small number of points at a time.  Our strategy, therefore, will be to fit a curve which relies neither on an ordering, partial or otherwise, of the boundary, nor does it rely on a consistent choice of orientation.\\

\par In two dimensions we have several options available to us.  One of the more well known techniques is that of B\`{e}zier curves.  There are several advantages to B\`{e}zier curves, notably we can work with them explicitly.  Additionally it is easy to patch together small curves and maintain differentiable continuity.  In our case we require at least two continuous derivatives so that we can compute curvature.  Recall that the curvature of a space curve $r(t)$ is given by
\[
\kappa = \left|\frac{d\mathbf{T}}{ds}\right|
\]
That is the derivative of the tangent vector with respect to arclength.  Since the tangent vector is already a derivative, this calculation will not be sensible if $r(t)$ is not at least twice continuously differentiable.  This will require us to patch together many B\`{e}zier curves if we want simpler derivatives, or risk a wildly oscillating curvature when we try to fit the entire boundary with a single curve.  The disadvantage then is that given $n$ control points we wish to fit, we require a degree $n-1$ polynomial for the B\`{e}zier curve.  This may cause us to overfit our data and risk experiencing Runge's effect of rapid oscillation.  One approach will be computationally too expensive, and the other will give an unnaturally large curvature.\\

\par In order to limit both the computational complexity and artifically large curvature we will leave B\`{e}zier curves and instead use B-splines, in particular we shall use nonuniform rational B-splines, hereafter NURBS.  The advantages to NURBS are many and the disadvantages few.  The first major advantage is that the basis splines have compact support.  
\begin{defn}
Consider an interval $[a,b]$ partitioned into $m$ subintervals $[t_i,t_{i+1}]$ by
\[
a= t_0 < t_1 < t_2 < \cdots < t_m = b
\]
then for some positive integer $k$ we define the basis functions $N_{i,k}(t)$ recursively by
\begin{eqnarray}
N_{i,1}(t) &=& \left\{\begin{array}{l l}
1 & t_i\le t < t_{i+1}\\
0 & \textrm{ elsewhere}
\end{array} \right. \\
N_{i,k}(t) & = & \frac{t-t_{i}}{t_{i+k-1}-t_{i}}N_{i,k-1}(t) + \frac{t_{i+k}-t}{t_{i+k}-t_{i+1}}N_{i+1,k-1}(t) \label{nurbs}
\end{eqnarray}

\end{defn}

\begin{lem}
The basis functions $N_{i,k}(t)$ form a partition of unity for all $k$
\end{lem}

\begin{proof}
This proof is a straight forward induction on the index $k$.
\end{proof}

It is not immediately obvious from the definition, but 
$N_{i,k}(t)$ is $C^{k-2}$ continuous.  Examining this a little, we see each incremental $k$ gives an extra power of $t$ as well as a convolution of lower order basis functions.  $N_{i,1}(t)$ is not globally continuous, $N_{i,2}(t)$ is simply continuous, but differentiable.\\

\par Given these basis functions for our spline we define a general NURBS curve for (spatial) points $\mathbf{p}_j$ and weights $w_j$ 

\begin{equation}
\vec{r}(t) = \frac{\sum_{j=0}^{m} \mathbf{p}_j w_j N_{j,k}(t)}{\sum_{j=0}^{m}  w_j N_{j,k}(t)}
\end{equation}

\begin{rem}
When the weights $w_j$ are uniform we reduce to a B-spline. Additionally, if the order $k$ of the basis functions is the same as the number of control points then this reduces to a B\'{e}zier curve.
\end{rem}

For our purposes, we require only $C^2$ continuity and thus from this point forward we shall work with the basis functions $N_{i,4}(t)$.\\

\par Extending this to surfaces sitting in three dimensions we parametrize both dimensions and get the parametric representation of the NURBS surface
\begin{equation}
r(u,v) = \frac{\sum_{i,j} \mathbf{p}_{i,j} w_{i,j} N_{i,k_1}(u)N_{j,k_2}(v)}{\sum_{i,j}  w_{i,j} N_{i,k_1}(u)N_{j,k_2}(v)}
\end{equation}

In many dimensions we have the parametrized manifold
\begin{equation}
r(u_1,\dots, u_n ) = \frac{\sum \mathbf{p}_{i_1,\dots,i_n} w_{i_1,\dots,i_n} N_{i_1,k_1}(u_1)\dots N_{i_n,k_n}(u_n)}{\sum  w_{i_1,\dots,i_n} N_{i_1,k_1}(u_1)\dots N_{i_n,k_n}(u_n)}
\end{equation}

We can see some of the advantages to NURBS manifolds over other splines.  First we can control the computational complexity by restricting the order of the basis functions.  That is that we can fit a bicubic spline to a surface in three dimensions, without the necessity of building $O(N)$ patches $N$ begin the number of points in a cluster.  Second, since our basis functions are compactly supported around our control points, we need not worry about orienting our cluster boundary nor must we concern ourselves with putting our control points in a ``proper" order.  Third, by adjusting our weights we can fit our curves and manifolds as precisely as we like.  Basically, increasing $w$ we push our curve closer to the actual points.  The advantage here, is that we have the ability to control noise in a measurable way.  Fourth, these manifolds are already parametrized which makes computations  in local coordinates feasible.  For surfaces, we can directly compute the first and second fundamental forms which give us our principal curvature explicitly.

\section{Calculating Curvature}

We have already mentioned that we don't wish to over fit our curves and surfaces, nor do we wish to under fit them.  Our goal then will be to assign to each cluster boundary a number which gives a good sense of its shape complexity.  For the moment we shall consider surfaces sitting in three dimensions.  In three dimensions there are two related, but distinct concepts called ``curvature."  The first is Gaussian curvature $K$ and the second is mean curvature $H$.  Each of these has it's own use and advantage, but we shall combine them to get a more effective curvature calculation.  On a space curve there is only only curvature, but in two dimension there are two ``principal" curvatures $\kappa_1$ and $\kappa_2$.  These indicate how much a surface is curving relative to its parametrization in the first or second coordinate.  The formulaic versions show:
\[
K = \kappa_1 \kappa_2 \textrm{ and } H = \frac{\kappa_1 + \kappa_2}{2}
\]

The problems we see computationally with each of these are the following: If one of the curvatures is zero, then $K$ is zero.  Consider then a very erratic curve in the plane cross the unit interval.
\[
C(t) \times [0,1]
\]

This will have $K=0$ uniformly.  In the authors' opinions, this is not a simple curve.  On the other hand, the mean curvature $H$ will deliver a uniformly zero score for a manifold which is locally hyperbolic everywhere.  A saddle for example is not a shape with zero complexity.  We therefore wish to use the average of the sum of squared principal curvatures:
\[
\frac{\kappa_1^2 + \kappa_2^2}{2} = 2 H^2 - K
\]

Additionally we know that shapes do not always have a constant curvature everywhere, so we wish to add up all curvature contributions.  Letting $\partial C$ be the boundary of our cluster and $d\sigma$ the surface area element we compute
\begin{equation}
Shape := \half \int_{\partial C} (\kappa_1^2 + \kappa_2^2) d\sigma
\end{equation}

More generally speaking, in $d$-dimensions we have $\kappa_1,\dots,\kappa_d$.  These sum of curvatures is a trace in two dimensions, and the product is a determinant.  In fact, in two dimensions given the first $I$ and second $II$ fundamental forms we can calculate principal curvatures by

\begin{equation}
\det(II - \kappa I) = 0 \label{principal}
\end{equation}

In $d$-dimensions this leads us to the idea that there are higher invariant forms of eigenvalues which we may compute by Newton polynomials.  Our higher dimensional shape score shall be defined by
\begin{equation}
Shape : = \frac{1}{d} \int_{\partial C} \left(\sum_{j=1}^{d} \kappa_j^2\right) dV
\end{equation}

\subsection{The Correct Curvature}

It is our duty therefore to pick a proper curvature which fits our surface well enough.  In this case we have many options.  We may set all weights to be identical, which gives a $B$-spline.  In pratice this is the most effective way to actually compute curvature.  However, for those who are more technically inclined we may consider two other options.  Maximizing curvature is not appropriate since there is no ``maximum" curvature.  We see this in the Runge effect.  If we fit a curve of degree 5000 to a cluster with 4999 points, we may approach infinite curvature.  Thus it will be in our best interest to minimze curvature relative to the boundary.  We have two major ways to do this.  The first is to give the absolute smallest possible curvature integral.  This seems to be a good measure, since spirals will have a constantly increasing curvature, while polygons will have a small curvature unless they are particularly small in which the curve will have to change quickly in a small space to fit the data well.  The second is to minimize the distance 
from the data to the curves.  This can be done via a least squares from points to lines.

\subsection{The Fundamental Forms}

For a parametrized surface $S(u,v)$ we define the following forms

\begin{eqnarray}
I &=& \begin{bmatrix}
E & F \\
F & G
\end{bmatrix}\\
E & = & S_u \cdot S_u \nonumber\\
F & = & S_u \cdot S_v \nonumber\\
G & = & S_v \cdot S_v \nonumber
\end{eqnarray}

\begin{eqnarray}
II & = & \begin{bmatrix}
L & M \\
M & N
\end{bmatrix}\\
L & = & S_{uu}\cdot \hat{n} \nonumber \\
M & = & S_{uv}\cdot \hat{n} \nonumber \\
N & = & S_{uu}\cdot \hat{n} \nonumber \\
\hat{n} & = & \frac{S_u \times S_v}{\| S_u\times S_v\|} \nonumber 
\end{eqnarray}

Given these forms called the first ($I$) and second ($II$) fundamental forms, respectively, we can calculate the Gaussian and mean curvatures as follows.

\begin{equation}
K = \frac{\det(II)}{\det(I)}
\end{equation}

\begin{equation}
H = tr((II)(I^{-1}))
\end{equation}

In coordinates we can reduce both the Gaussian and mean curvatures via

\begin{equation}
K = \frac{S_{uu}\cdot S_{vv} - S_{uv}^2}{\left(1 + S_u^2 + S_v^2\right)^2}
\end{equation}

\begin{equation}
H = \frac{(1+S_u^2)S_{vv} - 2S_u S_v S_{uv} + (1+ S_v^2)S_{vv}}{\left(1 + S_u^2 +S_v^2\right)^{3/2}}
\end{equation}

\subsection{B-splines and Curvature}

In practice we will use identically weighted NURBS curves, so that they are uniform and we need not worry about rational functions.  Calculations can, however, be made explicit for different weights.  \\

\par Recall that we only require two continuous derivatives and since we are using $w_{ij}=1$ we need to consider only the basis functions $N_{i4}(u)$ and $N_{j4}(v)$.  This reveals our parametrized surface to be

\begin{equation}
S(u,v) = \sum_{i,j} \mathbf{P}_{ij}N_{i4}(u)N_{j4}(v)
\end{equation}

Leading us to

\begin{eqnarray}
S_u & = & \sum_{i,j} \mathbf{P}_{ij}N'_{i4}(u)N_{j4}(v) \\
S_v & = & \sum_{i,j} \mathbf{P}_{ij}N_{i4}(u)N'_{j4}(v) \nonumber \\
S_{uu} & = & \sum_{i,j} \mathbf{P}_{ij}N''_{i4}(u)N_{j4}(v) \nonumber \\
S_{uv} & = & \sum_{i,j} \mathbf{P}_{ij}N'_{i4}(u)N'_{j4}(v) \nonumber \\
S_{vv} & = & \sum_{i,j} \mathbf{P}_{ij}N_{i4}(u)N''_{j4}(v) \nonumber 
\end{eqnarray}

At this point we see one of the big advantages of having a parametrized surface of separated variables.\\

Now let us consider the derivatives of our basis functions \ref{nurbs}.

\[
N_{i,k}(t)  = \frac{t-t_{i}}{t_{i+k-1}-t_{i}}N_{i,k-1}(t) + \frac{t_{i+k}-t}{t_{i+k}-t_{i+1}}N_{i+1,k-1}(t)
\]

This gives us a first derivative as

\begin{eqnarray}
N'_{ik}(t) & = & \frac{1}{t_{i+k-1}-t_i}N_{i,k-1}(t) + \frac{t - t_i}{t_{i+k-1}-t_i}N'_{i,k-1}(t) \nonumber \\
& + &  \frac{-1}{t_{i+k}-t_{i+1}}N_{i+1,k-1}(t) + \frac{t_{i+k}-t}{t_{i+k}-t_{i+1}} N'_{i+1,k-1}(t) 
\end{eqnarray}

\end{document}